\documentclass[nonacm,12pt]{acmart}
\AtBeginDocument{%
  \providecommand\BibTeX{{%
    \normalfont B\kern-0.5em{\scshape i\kern-0.25em b}\kern-0.8em\TeX}}}

\setcopyright{acmcopyright}
\copyrightyear{2020}
\acmYear{2020}
\acmDOI{}

\acmConference[]
\acmBooktitle{}

\usepackage{url}            
\usepackage{booktabs}       
\usepackage{amsfonts,amsthm,algorithm2e}

\newcommand{\E}{I\!\!E}
\newcommand{\gdense}{$\gamma$-clique \xspace}
\newcommand{\Aone}{\textsc{Detect}($\gamma ,\delta$)~}
\newcommand{\twocoreC}{2\hbox{-core}(C)}

\newsavebox{\fmbox}
\newcommand{\eps}{\varepsilon}
\newtheorem*{theorem-non}{Theorem}
\copyrightyear{2020}
\acmYear{2020}
\setcopyright{acmlicensed}\acmConference[FODS '20]{Proceedings of the 2020
ACM-IMS Foundations of Data Science Conference}{October 19--20,
2020}{Virtual Event, USA}
\acmBooktitle{Proceedings of the 2020 ACM-IMS Foundations of Data Science
Conference (FODS '20), October 19--20, 2020, Virtual Event, USA}
\acmPrice{15.00}
\acmDOI{10.1145/3412815.3416884}
\acmISBN{978-1-4503-8103-1/20/10}

\begin{document}
\fancyhead{}
\title{Large Very Dense Subgraphs in a Stream of Edges}

\author{Claire Mathieu}
\email{claire.mathieu@irif.fr}
\affiliation{
  \institution{CNRS and IRIF}
}
\author{Michel de Rougemont}
\email{mdr@irif.fr}
\affiliation{
 \institution{University Paris II and IRIF}
}

\renewcommand{\shortauthors}{}

\begin{abstract}
  We study the detection and the reconstruction of a large very dense subgraph in a social graph with $n$ nodes and $m$ edges given as a stream of edges, when the graph follows a power law degree distribution, in the regime when $m=O(n. \log n)$. A subgraph $S$ is very dense  if it has $\Omega(|S|^2)$ edges. We uniformly sample the edges with a Reservoir of size  $k=O(\sqrt{n}.\log n)$. Our detection algorithm checks whether the Reservoir has a giant component. We show that if the graph contains a  very dense subgraph of size 
$\Omega(\sqrt{n})$, then the detection algorithm is  almost surely correct. On the other hand, a random graph that follows a power law degree distribution almost surely has no large very  dense subgraph, and the detection algorithm is  almost surely correct.  We define a new model of random graphs which follow  a power law degree distribution and have large very dense subgraphs. We then show that on this class of random graphs we can reconstruct a good approximation of the very dense subgraph with high probability. We generalize these results to dynamic graphs defined by sliding windows in a stream of edges.
\end{abstract}


\ccsdesc[500]{Theory of computation~Streaming models}
\ccsdesc[500]{Theory of computation~Sketching and sampling}
\keywords{Dense subgraphs, clustering, streaming, probabilistic analysis, random graphs, approximation}

\settopmatter{printfolios=true}
\maketitle

This work was partially funded by the grant ANR-19-CE48-0016 from the French National Research Agency (ANR).\\\\
A preliminary version was presented at FODS 2020 (Foundations of Data Science) Conference.
\newpage
\tableofcontents
\newpage
\section{Introduction}

We study the efficient detection of  large dense subgraphs in social graphs, given as a stream of edges, by looking at only a small fraction of the stream. Our viewpoint comes from two constraints. First, given the massive size of social networks, our algorithms cannot store the graph in memory: the space used must be sublinear. Second, the dynamic feature corresponds to the evolution of the social network and our algorithms have to take a single pass over the stream. Given those two constraints, what kind of structure can be detected algorithmically?
 
{\bf Large dense subgraphs.} Social networks such as Twitter evolve dynamically, and dense subgraphs appear and disappear over time as interest in particular events grows and wanes. How can we detect large dense subgraphs efficiently?  The classical density is the ratio $\rho=|E[S]|/|S|$, where $E[S]$ is the set of edges internal to $S$.
In the case of a stream of edges, the approximation of dense subgraphs is well studied in \cite{B15,S15,E15,Mc15} and an $\Omega(n )$ space  lower bound is known \cite{Bah12}.   Social graphs define a specific regime for which we propose a streaming algorithm which uses $O(\sqrt{n}. \log n)$ space. Our density objective is however different.

\begin{definition}\label{gclique}
The {\em $(\gamma,\delta) $-large very dense subgraph} problem, where  $\gamma \leq 1$, takes as input a graph $G=(V,E)$ and decides whether there exists an induced subgraph $S\subseteq V$ such that $|S|>\delta\sqrt{n}$ and $|E[S]|>\gamma |S| (|S|-1)/2$.
\end{definition}

A {\em very dense subgraph}, also called a $\gamma$-clique, is a subset $S$ such that the condition $|E[S]|>\gamma |S| (|S|-1)/2$ holds. If you consider two random nodes of $S$, they are connected by an edge with probability $\gamma$.  
Observe that $\gamma$-cliques have the following nice structure, which does not hold for the usual measure $|E(S)|/|S|$: if $S$ is a \gdense then for any $2\leq i\leq |S|$, there exists a set of size $i$ that is a $\gamma$-clique. In fact, a random subset of $S$ has this property on average.
The  {\em $(\gamma,\delta)$-large very dense subgraph} problem is NP-hard and hard to approximate~\cite{Hastad:1996}, as it contains the maximum clique problem as the special case when $\gamma=1$. This leads us to use a new notion of approximation, adapted to a specific distribution of the inputs.


{\bf Degree distribution and the Configuration model.}
A scale-free network is a network whose degree distribution asymptotically follows a power law: the fraction of nodes with degree $d$ is proportional to $d^{-O(1)}$ for $d$ tending to infinity. Many real-world networks are thought to be scale-free. In this paper, we focus on social network models with a fixed degree sequence following a power law $d^{-2}$, which is such that the total number of edges is $m=O(n. \log n)$. 
The configuration model takes a feasible degree sequence and generates a multigraph with that degree sequence. Let 
 $\mu$ denote the distribution of simple graphs
obtained with the configuration model applied to the power law degree distribution $d^{-2}$. (We defer the discussion of how to get a simple graph from the multigraph to section~\ref{uniform}).
The configuration model is a standard model used already in sociology in 1938 in a directed version~\cite{morenojennings}, and also for modeling the World Wide Web~\cite{newmanstrogatzwatts}, public opinion formation~\cite{wattsdodds}, etc. 
There are other generative models  such as the Preferential Attachment model~\cite{b99}, the Copy model~\cite{Kleinberg1999}, and many others, see~\cite{kumar2000,AlbertBarabasi2000,ACL00,N10} for example.  

{\bf One-sided-stochastic approximation.} We relax the definition of a randomized algorithm $A$, with a one-sided approximation, where for Yes instances we consider the worst case, but for No instances we  only consider random inputs for $\mu$.  Indeed, a standard way to relax a decision problem would be to design an $\epsilon$-tester that decides whether the input is a Yes instance or is $\epsilon$-far from Yes instances, i.e. within edit distance at least $\epsilon m$ from any Yes instance; but for our problem, no graph is $\epsilon$-far from Yes instances, so the notion of $\epsilon$-tester is not the right one here.  

A {\em one-sided stochastic 
randomized} Reservoir algorithm $A$ for a language $L$ satisfies the following two conditions:
\begin{itemize}
\item For all $x\in L$, $Prob_{\Omega} [ A(x)~ accepts] \geq 1-\eps$

\item If $x\not\in L$ is drawn from $\mu$, $Prob_{\mu\times\Omega} [ A(x)~ rejects] \geq 1-\eps$

\end{itemize}
where $\Omega$  is the set of possible Reservoirs (subsets of edges with prescribed cardinality), with a uniform distribution.
We design a one-sided randomized streaming algorithm for the $(\gamma,\delta) $-large very dense subgraph problem and believe that this approximation for a distribution of inputs can also be useful in some other contexts.

{\bf Static results.}  We study how to decide the presence of large very dense subgraphs with a Reservoir  sampling \cite{V85} and how to reconstruct such a very dense subgraph from the samples\footnote{In this paper, whenever we speak of a ''dense subgraph'', we mean a {\em large} very dense subgraph, where the minimum size to qualify as large is specified in the theorem statements.}.
 We study the existence of giant components in the Reservoir using random graph techniques, adapted to graphs with this degree distribution. Indeed, the Reservoir is a random uniform sample of a power law graph, whereas the Erd\"os-Renyi model is a random uniform sample of the complete graph.  If the graph has a very dense subgraph $S$ of size $\Omega(\sqrt{n}  )$, then a Reservoir of size $\Omega(\sqrt{n}.\log n)$ has a giant component (Theorem~\ref{tConcentratedSize}). In order to detect the existence of a very dense subgraph, our first algorithm (Algorithm \Aone) simply checks whether there is a large enough connected component in the Reservoir.  The analysis relies on a Theorem by Molloy and Reed on asymptotic sequences, but there are additional difficulties here, due to the fact that our sequences are not deterministic but random.

We then analyze the situation when the graph does not have a large very dense subgraph: we take the configuration model of random graphs
which  follow a power law degree distribution (\cite{N10}, see also next section). We observe that in this case there is no giant component in the Reservoir (Lemma~\ref{lemma:noclique}) and that Algorithm \Aone is correct (Theorem~\ref{t1}).  In Corollary~\ref{corollary:lowerbound} we show an $\Omega(\sqrt{n})$ space lower bound. 

Given a graph that has a very dense subgraph, how can we not only detect its existence but also reconstruct it? We  propose a simple algorithm (Algorithm 2) that uses the $2$-core of the largest connected component of the Reservoir. We define a variant of the configuration model, giving a preference to edges inside a certain subset of the vertices, so that the graph contains a very dense subgraph,  and show that in this case Algorithm 2  reconstructs an approximation of the very dense subgraph (Theorem \ref{thm:estimation}.)

{\bf Social networks dynamics and dynamic results.} We consider "bursty" windows when the number of edges in each window varies. Uniform sampling in a window is the Reservoir sampling \cite{V85} and it has been generalized to overlapping windows in \cite{B02,B09}.
We  turn to the dynamic case with sliding windows.  We show in Corollary \ref{t0d} that we can detect the existence of a large very dense subgraph by generalizing Algorithm \Aone (Algorithm 3).  The configuration model can be generalized to
dynamic random graphs (Section~\ref{section:dynamic}), and then none of the successive Reservoirs  have a giant component,  Corollary \ref{t0duni}. In the 
concentrated case for some time interval $\Delta$, the random graphs during that interval  have a very dense subgraph $S$ and Algorithm 4 approximates $S$, Corollary \ref{t1d}.  

{\bf Plan of the paper.} In section 2, we review the large dense subgraphs, Reservoir sampling for dynamic graphs and random graphs. In section 3 we describe random graphs which follow a power law degree distribution with or without  large very dense subgraph. In section 4 we present the Algorithm \Aone and its analysis  on positive and random instances.
 In section 5, we study the space lower bounds. In section 6, we show how to reconstruct a good approximation of large very dense subgraph from the samples. In section 7, we generalize the approach to dynamic graphs defined by sliding windows.

\section{Preliminaries}
Throughout the paper, we  ignore the rounding of our parameters to the nearest integer, when it only has an impact on second-order terms. For example, we look for a very dense subgraph with at least $\delta \sqrt{n}$ nodes.
We approximate the existence of a very dense subgraph in the regime where $m=O(n.log n)$, observed in social graphs. 

\subsection{Large dense subgraphs}

There are several definitions of a cluster in a graph  \cite{Ag10} and our definition \ref{gclique} assumes a large $\gamma$-clique, as setting the value of the parameter $\gamma$ to 1 corresponds to a clique for $S$. 
Studies of the Web graph have previously used large bipartite cliques as a defining marker of Web communities~\cite{Kleinberg1999}. This differs from a common definition of dense subgraph according to which 
 the average degree within the subgraph, $\rho(S)=E(S)/|S|$, must be greater than some fixed threshold $\lambda$.
A classical NP-hard~\cite{K09} optimization problem is to approximate   $\rho^*=Max_S \{ \rho(S)\}$ and to find possible witnesses $S^*$ for  $\rho^*$, in particular when the graph is given as a stream of edges \cite{Mc15,B15,E15,Mc15}.  If the entire graph is known, there are several classical  techniques to find large   \gdense subgraphs. 
In this context,  it is hard   to approximately detect subgraphs $S$ with large value of $|E(S)|/|S|$:
 an $\Omega(n)$ space lower bound is known \cite{Bah12} based on the   multiparty Disjointness problem 
\cite{Bar04,Cha03}. As we are interested in large  dense subgraphs,  i.e. $|S| \geq \Omega(\sqrt{n})$, we observe in Corollary~\ref{lowerbound} that in this case  the space lower bound is reduced to $\Omega(\sqrt{n})$. 

For a large very dense subgraph, we show in Corollary~\ref{corollary:lowerbound} that the same space lower bound  $\Omega(\sqrt{n})$ applies. For a comparison, our algorithm uses space $O(\sqrt{n}.\log n)$.

\subsection{Sampling social graphs and dynamic graphs}

Social graphs have a specific structure: a specific degree distribution  (a power law), and some very dense  subgraphs.
In our framework, we do not store the entire graph as we sample the edges and will only approximate these very dense subgraphs.
 For example, imagine that the entire graph is a complete graph $S$. Each edge  is selected in the Reservoir with probability $\frac{k}{m}$, so the Reservoir follows the Erd\"{o}s-Renyi model $G(n,p)$ (see next section) where $n=|S|$ and $p=\frac{k}{m}$. It is well-known that there is a phase transition  at $p=1/n$, for the emergence of a giant component. We first  show that for a   \gdense subgraph, a giant component emerges at $p=1/\gamma.n$. We then study sufficient conditions to observe a giant component in the Reservoir.

Dynamic graphs and  models of densification  in social graphs have been studied in \cite{L05}. Dynamic algorithms approaches are presented in \cite{G10}. We consider sliding windows defined 
by two parameters: the time length $\tau$, and  the time  step $\lambda <\tau$ such that 
$\lambda$ divides $ \tau$.
In a stream of edges  $e_1, e_2,....e_i...$, each edge has a timestamp. Let $t_1=\tau$ and  $t_i=\tau+\lambda.(i-1)$ for $i>1 $. Let $G_i$ denote the graph defined by the edges whose timestamps are  in the time interval $[t_i -\tau, t_i]$. We write $G(t)$ for this sequence. The graphs $G_{i+1}$ and $G_{i}$ share many edges: old edges of $G_i$ are removed and new edges are added to  $G_{i+1}$. 

A Reservoir sampling \cite{V85} reads a stream of edges $e_1, e_2,...e_i,...e_n$ and selects
$k$  edges uniformly ($k \leq n$), i.e. each edge is chosen with probability $\frac{k}{m}$.  It first initializes the Reservoir to $e_1, e_2,....e_k$. Then for each $i$ such that $k <i \leq n$, it selects $e_i$ with probability $k/i$. If $e_i$ is selected, it replaces a random element of the Reservoir, selected with probability $1/k$.
In a dynamic window $w_i$, we assume we can maintain $k$ uniform samples in each  Reservoir $R_i$, using  techniques  presented in \cite{B02, B09}.

\subsection{Random graphs}

\subsubsection{Erd\"os-Renyi graphs}\label{rg}
The classical Erd\"{o}s-Renyi model $G(n,p)$ \cite{E60}, generates  random graphs with $n$ nodes and edges are taken independently with probability $p$ where $0<p<1$.  The degree distribution is close to a Gaussian centered on $n.p$.  
 A {\em giant component} is a connected component with $\Theta(n)$ vertices.  A classical study is to find sufficient conditions so that the random graph has a {\em giant component}. In  the Erd\"{o}s-Renyi model $G(n,p)$, it requires that $p>1/n$. If  $p=c/n$ with $c>1$, we can be more precise on the size of the giant component $C$ and the size of the $\twocoreC$. 
\begin{definition}
The $\twocoreC$ is obtained from $C$ by removing nodes of degree 1, iteratively.
\end{definition}

\begin{theorem}  [{\bf 6 }  from \cite{PW05}] \label{er:bound}
Let $c>1$ be fixed. Consider the Erd\"{o}s-Renyi model $G(n,p)$ with $p=c/n$. Let $C$ denote the largest connected component and $\twocoreC$ its 2-core. 
There exists $b,t$ such that $t.e^{-t}=c.e^{-c}$ and $b=1-t/c$ such that:
$$lim_{n \rightarrow \infty} \E|C|/n = b$$
$$lim_{n \rightarrow \infty} \E|\twocoreC|/n = b\cdot(1-t)$$
In addition, variables  $|C|$ and 
$ |\twocoreC|$ are both Gaussian in the limit. 
\end{theorem}
We will use this result in section \ref{Sestimate}. A similar result 
(Theorem 5.1) is given in \cite{DLP14}.
When $c$ is a large constant (tending slowly to infinity), $t$ is close to 0, $\Theta( c e^{-c})$, and so $b=1-\Theta(e^{-c})$ tends to 1, therefore both $C$ and $\twocoreC$ have size tending to $n$ with high probability.
We use the following result which generalizes the Erd\"{o}s-Renyi condition to observe a giant component in a graph $G$.  Let $d_u$ be the degree of a node $u$.

When we write that a property holds \emph{almost surely}, we mean that the probability is at least $1-o(1)$ as $n$ goes to infinity. \footnote{In some other contexts, it is sometimes called {\em asymptotically almost surely}. }

\begin{theorem}  [ \cite{C09}] \label{gamma}
  Let ${\bar d}=    \frac{\sum_u d_u^2}{\sum_u d_u}$. Let $\epsilon>0$ be fixed. Let $\widehat{G}$ denote the random subgraph of $G$ obtained by choosing each edge with probability $p$. If  $p>(1+\eps)/{\bar d}$ then  $\widehat{G}$  almost surely has a giant component.
\end{theorem}

A direct consequence which we will use in section \ref{aa} is the following lemma:

\begin{lemma}\label{delta}
Let $\epsilon>0$ be fixed. If $S$ is a   \gdense and $p> \frac{1+\epsilon}{\gamma.|S|}$ then the random subgraph $\widehat{S}$ obtained in $G(|S| ,p)$,  almost surely has a giant component.
\end{lemma}
\begin{proof}
Notice that ${\bar d}=    \frac{\sum_u d_u^2}{\sum_u d_u} \geq \frac{\sum_u \gamma^2. |S|^2}{\sum_u d_u} \geq \frac{ \gamma^2. |S|^3}{\gamma. |S|^2}= \gamma.|S| $. In the first inequality
 $ \sum_u d_u^2$ is minimized when the degrees are uniform.  The second inequality uses the 
definition of a   \gdense: $ \sum_u d_u \geq \gamma. |S|^2$.
Hence $$ \frac{ 1}{{\bar d}}\leq  \frac{ 1}{\gamma.|S| }$$

We conclude that  if $p> \frac{1+\eps}{\gamma.|S|} \geq \frac{ 1+\eps}{{\bar d}}$ and by the 
Theorem \ref{gamma}, $\widehat{S}$  almost surely has a giant component.

\end{proof}

\subsubsection{Power-law graphs}\label{rgpl}

Most of the social graphs have a degree distribution close to a power law (such as a Zipfian distribution distribution where $Prob[d=j]=c/j^2$).  The preferential attachment or the configuration model
\cite{N10} provide models where the degree distribution follows such a power law. In the configuration model, the degree distribution can be an arbitrary distribution ${\mathcal D}$: given $n$ nodes, we  fix  the number   $d_1, d_2, ...d_{max} $ of nodes of degree $1,2,...$ where $d_{max}$  is the maximum degree.

\begin{definition}\label{lemma:Zipf}
Given $n$, let $c>0$ be such that $\sum_{i\geq 1}\lfloor nc/i^2\rfloor =n$ ($c$ is approximately $6/\pi^2$). In the {\em Zipf degree sequence}, there are $d_i(n)=\lfloor nc/i^2\rfloor$ nodes of degree $i$ for each $i> 1$. The number of nodes of degree 1 is either $\lfloor nc\rfloor$ or $\lfloor nc\rfloor-1$, chosen so that the sum of degrees is even; in the former case, all nodes have degree at least 1;  in the latter case, there is one node of degree 0. 
\end{definition}
\begin{lemma}\label{lemma:feasible}
For all $n$, the Zipf degree sequence is feasible: there exists a graph with that degree sequence.
\end{lemma}
\begin{proof}
The definition of the degree sequence guarantees that the sum of degrees is even, and it is easy to see that the sequence then satisfies the Erd\"os-Gallai condition~\cite{EG60}.
\end{proof}

\begin{lemma}\label{lemma:basic}
The Zipf degree sequence satisfies the following elementary properties.
\begin{itemize}
\item
The maximum degree is $dmax=\sqrt{cn}$.
\item
The total number $m$ of edges satisfies
$$|m-\frac{cn\ln(cn)}{4}|\leq \frac{cn}{2}$$
and the average vertex degree satisfies
$$|\E(d_u)-\frac{c\ln(cn)}{2}|\leq {c}$$
\item
Let  $d=\delta\sqrt{n}=dmax/\sqrt{a}$ with $a=c/\delta^2$. The number of nodes of degree greater than or equal to $\delta\sqrt{n}$ is
$$\sqrt{cn}[\sum_1^{a-1}\frac{1}{\sqrt{i}}-\frac{a-1}{\sqrt{a}}]+O(1).$$
\end{itemize}
\end{lemma}


\section{Graphs with and without large very dense subgraphs}
In this section, we consider random graphs with a power-law degree distribution, first without large very dense subgraphs and then with such a subgraph.
\subsection{Uniform configuration model} \label{uniform} 

\begin{definition}\label{definition:configuration} 
The  \emph{configuration model}  generates a multigraph  from a given degree sequence. 
The model  constructs a  uniform random perfect matching $\pi$ between stubs as a symmetric permutation without fixed points: if $\pi(i)=j$, meaning that stub $i$ maps to stub $j$, then 
$j \neq i$ and  $\pi(j)=i$. 
To construct $\pi$, we greedily match stubs at random. 
We obtain a multigraph.\footnote{ Note that the distribution of the multigraph thus obtained is not uniform in general.  We can obtain simple graphs, i.e. without self-loops  or multi-edges by rejection sampling, and they satisfy the same properties. Since, for each simple graphs with that degree distribution, the number of executions leading to that simple graph is the same, the restriction of the distribution to simple graphs yielded by rejection sampling is uniform.  \cite{MR98}. 
 }
\end{definition}

For a given degree sequence, let $\mathcal D$ denote the random variable equal to the degree of a uniform random node, and $G$ denote  the multigraph obtained from the degree sequence according to the configuration model.

 \begin{definition}\label{definition:wellbehaved}
Let $d_i(n)$ denotes the number of vertices of degree $i$ in an $n$-vertex graph. We define a well-behaved asymptotic degree distribution $(d_i(n))_{i,n}$:
\begin{enumerate}
\item It is feasible: for each $n$, there exists a graph with degree distribution $(d_i(n))_i$.
\item It is smooth: for each $i$, $\ell_i=\lim_{n\rightarrow\infty} d_i(n)/n$ exists.
\item The convergence of $ i(i-2) d_i(n)/n$  to its limit $\ell_i$ is uniform: $\forall \epsilon  ~\exists N ~\forall n>N ~\forall i~ |i(i-2)d_i(n)/n - \ell_i|<\epsilon$.
\item $L({\mathcal D})=\lim_{n\rightarrow\infty} \sum_i  i(i-2) d_i(n)/n$ exists and the convergence is uniform: 

$\forall\epsilon~\exists i_0 ~\exists N ~\forall n>N~
|\sum_{i\leq i_0}  i(i-2) d_i(n)/n-L({\mathcal D})|<\epsilon$.
\end{enumerate}
\end{definition}

\begin{theorem} (1(b) from Molloy-Reed \cite{MR98}) \label{mr:bound})
Consider an asymptotic degree sequence such that:
\begin{enumerate}
 \item
 the asymptotic degree sequence is well-behaved, 
\item
 $Q({\mathcal D})=
 \sum_i  (i^2-2i) \ell_i$ is less than a constant less than 0, 
 \item
 the maximum vertex degree is less than $n^{1/9}$, and
 \item the average vertex degree is $O(1)$.
 \end{enumerate}
Then the following holds almost surely  in the configuration model: 
the largest connected component of $G$ has size at most $Bn^{1/4}$ for some constant $B$ depending on $Q({\mathcal D})$; no connected component of $G$ has more than one cycle; and  there are at most $2Bn^{1/4}$  cycles. 
\end{theorem}

As noted in \cite{MR98}, for a well-behaved sequence we have 
$$L({\mathcal D})=\lim_{n\rightarrow\infty} \sum_i  i(i-2) d_i(n)/n= \sum_i  i(i-2) \lim_{n\rightarrow\infty}d_i(n)/n= Q({\mathcal D})$$

We will apply Theorem~\ref{mr:bound} to the Reservoir $R$ to bound the size of its largest connected component. 

First we prove a simple property of the uniform configuration model.

\begin{lemma}\label{lemma:noclique}
 In the uniform configuration model, $G$ almost surely does not have  a \gdense  of size $\Omega(\sqrt{n})$.
\end{lemma}
\begin{proof} 
Assume, for a contradiction, that there exists a set $S$ of size $c\sqrt{n}/a$ which is $\gamma$-dense. Then there exists a subset $S_1\subseteq S$ of size at least $c_1\sqrt{n}$ and that is at least $c'_1\gamma$-dense, and whose vertices have minimum degree $c''_1\sqrt{n}$. Let $A$ consist of all vertices of degree of $G$ at least $c''_1\sqrt{n}$. Note that $A$ is independent of $S$ (thus sparing the need to do some union bound) and that  $|A|\leq c_2\sqrt{n}$. Then $A$ itself must be $c_2\gamma$-dense. 

Let $E(A)$ the set of possible internal edges in $A$:  then  $|E(A)|=\theta(n)$, whereas we have $m=\theta(n.\log n)$ edges. The probability that a random edge, created by 
 the random stub-matching algorithm,  is in $E(A)$ is $\theta(n/n.\log n)=\theta(1/\log n)$.  The expected number of edges internal to $A$ is at most $\theta(\sqrt{n}/\log n)$ and the expected density $|E(A)|/|A|=\theta(1/\log n) $, i.e.  $o(1)$. By Markov's inequality, the probability that $A$ is $c_2\gamma$-dense is $O(1/\log n)$, hence we obtain a contradiction.
\end{proof}

 \subsection{Models with a large very dense subgraph}\label{subsection:uniform}
The configuration model  generates a random graph which follows a power law degree distribution, as  explained in section \ref{rgpl}, with a 
 uniform matching between the stubs.  With a  power law degree distribution,   the random graph $G$ has a giant component almost surely. We show in Theorem \ref{t1} that the Reservoir almost surely does not have a giant component and that $G$ does not have a large very dense subgraph.


We now define a {\bf concentrated model}, a stochastic power law model that defines a graph that contains a large very dense subgraph. 
In Theorem~\ref{thm:estimation}, we will prove that in that case, there is an algorithm, Algorithm 2, that not only decides but also reconstructs (approximately) the underlying very dense subgraph.  

Our model for a random graph with a very dense subgraph (concentrated model) is as follows. Let $\delta\leq \sqrt{c}/2$ and $0<\gamma \leq 1$ be fixed.  We now construct   a graph with a \gdense.

\begin{itemize}
\item We attach  $d$ stubs (half-edges)  to each node $v$  of degree $d$ of the distribution ${\mathcal D}$,
\item  Let $S$ be a set of $\delta\sqrt{n}$ 
 nodes chosen arbitrarily among all nodes that have degree at least $\delta\sqrt{n}$. 
\item  For each vertex $u$ of $S$, we mark $\delta\sqrt{n}-1$ stubs of $u$ at random. For each pair of marked stubs between different vertices, we put in the graph the edge between those stubs independently and with probability $\gamma$, creating a random graph over $S$ distributed as the Erd\"os-Renyi model $G(|S|,\gamma)$;
\item  The remaining  unmarked stubs of $V$ and the marked stubs not chosen in (1) are matched uniformly. 

\end{itemize}
By Lemma~\ref{lemma:basic},  for $\delta\leq \sqrt{c}/2$ there are at least  $\delta\sqrt{n}$ nodes of degree at least $\delta\sqrt{n}$, so there exist such subsets $S$.
Let $0<\gamma \leq 1$ be fixed.  
Notice that  the second procedure may add some edges in $S$.

\section{Existence of  a large very dense subgraph}


  Let $C$ be the largest  connected component of the Reservoir of size $k=\Theta(\alpha\sqrt{n}\log n)$.  
In order to decide the graph property $P$: {\em  is there is a  \gdense  of size greater then $\delta.\sqrt{n}$?} Consider this simple algorithm, where $\alpha > 1/\gamma\delta$ is an auxiliary parameter.\\

\begin{algorithm}\label{algo:A1}
{\bf Algorithm \Aone }

\SetAlgoVlined
\KwIn{a stream of edges of a graph $G$.}
 \KwOut{
  whether $G$ contains a large very dense subgraph }
 Let $\alpha=\Theta(1/(\gamma\delta))$. \\
$\bullet$ Use Reservoir sampling to maintain a reservoir $R$ of size $k=\Theta(c.\alpha \sqrt{n}\log n)/4)$.\\
$\bullet$Let $C$ denote the vertices of the largest connected component of $R$.

$\bullet$ Accept iff $|C|\geq \Theta(n^{1/8}\log^2 n)$ 
\end{algorithm}

The bound $n^{1/8}\log^2 n$ is a direct application of the Molloy-Reed Theorem \ref{mr:bound}.
The  sampling rate  $\frac{k}{m}=\frac{c.\alpha \sqrt{n}\log n/4}{cn.\log n/4}=\frac{\alpha}{\sqrt{n}}$.
We show the correctness of this algorithm in two steps. In section \ref{uni}, we show in Theorem \ref{t1} that Algorithm \Aone rejects almost surely in the uniform case. In the next section, we show in Theorem \ref{tConcentratedSize} that Algorithm \Aone accepts almost surely if there is a large   \gdense.


\subsection{Analysis when there is a large very dense subgraph}\label{aa}
In this section we analyze the size of the largest connected component $C$ of the Reservoir used for the detection of a very dense subgraph (Algorithm \Aone).
The following theorem formalizes the fact that Algorithm \Aone is correct with high probability on any graph  that contains a large   \gdense and has $m=cn.\log n/4$ edges.
\begin{theorem}\label{tConcentratedSize}
Assume that $G$ contains a   \gdense on $S$ where  $|S| > \delta.\sqrt{n}$ and has $m=cn.\log n/4$ edges. If  
$\alpha >  \frac{ (1+\epsilon)}{\gamma .\delta}$, then Algorithm \Aone Accepts almost surely.
\end{theorem}
\begin{proof}
If $S$ is a   \gdense and $G$ has  $m=cn.\log n/4$ edges:
$$\frac{k}{m}=\frac{\alpha}{\sqrt{n} }>  \frac{ (1+\epsilon)}{\gamma .|S|}$$
By Lemma  \ref{delta} there is  a giant component as
$$\frac{k}{m} > \frac{ (1+\epsilon)}{\gamma .|S|}$$
Hence Algorithm \Aone accepts almost surely.
\end{proof}

The above theorem shows that on positive instances, Algorithm \Aone is almost surely correct. 
What about negative instances? We observe that there exists an input graph $G$ that does not have a   \gdense of size strictly greater than $\epsilon\sqrt{n}$, yet which Algorithm \Aone (incorrectly) accepts. $G$ consists of a clique $K$ of size $\epsilon\sqrt{n}$ and of a path of size $n-|K|$. With high probability, the Reservoir contains a component of size at least $90\%$ of $K$, and will therefore accept, incorrectly.

However this input is somewhat pathological. In Theorem~\ref{t1}, we will prove that, assuming that $G$ is drawn from the following stochastic power  law model, the algorithm is correct on $G$ with high probability.

\subsection{Analysis in the uniform case}\label{uni}

\begin{theorem}\label{t1}
 In the uniform configuration model (subsection~\ref{rgpl}), Algorithm \Aone Rejects almost surely.
\end{theorem}

The Reservoir $R$ is constructed by a 3-step process, which we call the \emph{Configuration-first} process:
\begin{enumerate}
\item
Take the Zipf degree sequence (Definition~\ref{lemma:Zipf}).
\item
Use the configuration model (Definition~\ref{definition:configuration}) to generate a multigraph $G$ from the Zipf degree sequence
\item
Take a random uniform sample (without replacement) of $k$ edges from $G$ to define the Reservoir $R$.
\end{enumerate}
Instead, we will analyze the following process, which we call the \emph{Configuration-last} process:
 \begin{enumerate}
\item
Take the Zipf degree distribution (Definition~\ref{lemma:Zipf}) as the degree sequence of the overall graph $G$.
\item
Take a random uniform sample (without replacement) of $2k$ stubs, determining the degree sequence in $R$.
\item 
Use the Configuration model (Definition~\ref{definition:configuration}) to generate a random multigraph $R$ from that degree sequence.
\end{enumerate}

\begin{lemma}\label{lemma:equivalent}
The above two processes give the same distribution for multigraph $R$.
\end{lemma}
\begin{proof}
The configuration model is simply a greedy matching of stubs.
\end{proof}

\begin{proof} (of Theorem~\ref{t1})
By Lemma~\ref{lemma:equivalent}, we will analyze the second process for generating $R$.
The plan is to apply Theorem~\ref{mr:bound} to the configuration model generated from the degree sequence of $R$.
One difficulty is that that degree sequence is not deterministic but random. 
Theorem~\ref{mr:bound} assumes a deterministic degree sequence for each $n$. In our setting, we have a distribution of degree sequences for each $n$. 

We analyze the properties of the random degree sequence in $R$, so as to prove that with probability $1-o(1)$ it satisfies the assumptions of Theorem~\ref{mr:bound}.  We define one degree sequence for each $n$. We will prove that this asymptotic sequence satisfies the assumptions of Theorem~\ref{mr:bound}.

The proof relies on a series of technical Lemmas, that are deferred to the next subsections. 
The degree distribution of the Reservoir is well-behaved: the first property (feasibility) holds by definition. The second property (smoothness) holds by Lemma~\ref{x1a} for $i=1$ and by Lemma~\ref{xi} for $i\geq 2$. The third property (uniform convergence) also holds by Lemma~\ref{uniformconvergence}. The fourth property (uniform limit) holds by Lemma~\ref{uniformlimit}. Thus the degree sequence is well-behaved. As to the other assumptions of Theorem~\ref{mr:bound}, the second one ($Q$ negative) follows from the previous Lemmas, the third one (maximum degree) from Lemma~\ref{maxvertexdegreeR}, and the fourth one (average degree) from Lemma~\ref{averagedegreeR}. 

By Theorem~\ref{mr:bound}, we then have that, with probability $1-o(1)$, the largest connected component of $R$ has size 
$O(n_R^{1/4}\log{n_R})$. Since $n_R=O(\sqrt{n}\log n)$, this is $O(n^{1/8}\ln^{5/4}{n})$, and then Algorithm \Aone rejects.
\end{proof}

\subsubsection{Degree distribution in the Reservoir}

Let $\mathcal{D}_R$ denote the distribution of degrees in $R$.
Let $N_R$ denote the number of nodes spanned by the edges of $R$, with expectation $n_R=\E[N_R],$ 
and let $X_i$ be the random variable equal to the number of nodes of degree $i$ in $R$, with expectation  
$x_i=\E[ X_i   ]$,
where the expectation is over the Configuration-last process.
We will use the following basic fact: \begin{lemma}\label{fact:basic}
If $n\geq 0$ and $0\leq x\leq 1$ then $(1-x)^n\geq 1-nx$.
\end{lemma}

\begin{lemma}\label{x1a}
In the Reservoir $R$:
\begin{itemize}
\item
$X_1\leq N_R\leq 2k$.
\item 
$\E (X_1)=x_1\geq 2k(1- O(\frac{\ln\ln n}{\ln n}) ).$
\item
Let $\eta>0$. Then, with probability at least $1-\eta$ we have:  $ X_1\geq 2k(1- O(\frac{\ln\ln n}{\eta.\ln n}))$.
\item
The limit in probability of the ratio $X_1/N_R$ exists and is equal to $\ell_1=1$: that is, $plim_{n \rightarrow \infty} {X_1}/{N_R} =1$.
\end{itemize}
\end{lemma}

\begin{proof}
The first statement is obvious: the number of vertices of degree 1 is at most the toal number $N_R$ of vertices spanned by the $k$ edges of $R$, which is at most $2k$.

For the second statement, we start from an exact expression for the expected number of nodes of degree 1  in $R$. A vertex that has degree $j$ in $G$ has degree 1 in $R$ if and only if the Reservoir picks exactly one of its $j$ edges, which has probability $j\cdot (\alpha/\sqrt{n})(1-\alpha/\sqrt{n})^{j-1}$. Since there are $\lfloor cn/j^2\rfloor$ vertices of degree $j$ in $G$, we can write:
$$x_1= \sum_{j=1}^{\sqrt{c.n}} ~\lfloor \frac{c.n}{j^2}\rfloor \cdot j \frac{\alpha}{\sqrt{n}} (1-\frac{\alpha}{\sqrt{n}})^{j-1}.$$
Let  $\epsilon=\frac{2\ln\ln n }{\alpha\sqrt{c}\ln(n)}$. We use Lemma~\ref{fact:basic} to write:
$$(1-\frac{\alpha}{\sqrt{n}})^{j-1}\geq 
\left\{
\begin{array}{ll}
1-\epsilon{\alpha}{\sqrt{c}} & \hbox{if } j\leq \epsilon\sqrt{cn}\\
0 & \hbox{otherwise.}
\end{array}
\right. $$
Thus 
$$x_1\geq (1-\epsilon{\alpha}{\sqrt{c}}) \sum_{j=1}^{\epsilon \sqrt{c.n}} ~\lfloor \frac{c.n}{j^2}\rfloor \cdot j \frac{\alpha}{\sqrt{n}} .$$
On the other hand, since $2k$ equals the expected sum of the degrees of vertices in $R$, and a vertex of degree $j$ in $G$ has expected degree $j\cdot \alpha / \sqrt{n}$ in $R$, we can therefore write
$$2k=\sum_{j=1}^{\sqrt{c.n}} ~\lfloor \frac{c.n}{j^2}\rfloor \cdot j \frac{\alpha}{\sqrt{n}} $$ 
and
$$x_1\geq (1-\epsilon{\alpha}{\sqrt{c}}) \left(  
2k - \sum_{j=\epsilon\sqrt{cn}+1 }^{\sqrt{c.n}} ~\lfloor \frac{c.n}{j^2}\rfloor \cdot j \frac{\alpha}{\sqrt{n}} 
\right).$$
Now we can bound the last term by
$$
\sum_{j=\epsilon\sqrt{cn}+1 }^{\sqrt{c.n}} ~\lfloor \frac{c.n}{j^2}\rfloor \cdot j \frac{\alpha}{\sqrt{n}}  \leq
{c\sqrt{n} \alpha } \int_{\epsilon\sqrt{cn} }^{\sqrt{c.n}}\frac{dt}{t} = {c\sqrt{n} \alpha }\ln (1/\epsilon).
$$
Recall that $k=m\alpha/\sqrt{n}$ and $m\sim cn\ln(n)/4$, so that $2k\sim c\sqrt{n}\alpha \ln(n)/2$. We obtain:
$$x_1\geq 2k (1-2\frac{\ln(1/\epsilon)}{\ln(n)})  (1-\epsilon{\alpha}{\sqrt{c}}).$$
Substituting the value of $\epsilon$, we obtain $x_1\geq 2k(1-O(\frac{\ln\ln n}{\ln n}))$.

For the third statement, we use Markov's inequality on $Y_1=2k-X_1$. We have: $Y_1\geq 0$, and $\E (Y_1)=2k-x_1 \leq 2kO(\frac{\ln\ln n}{\ln n}) $ by Lemma~\ref{x1a}, so $\Pr (Y_1>\E(Y_1)/\eta)<\eta$, hence the result.

For the fourth statement, observe that $X_1/N_R\leq 1$ always. Let  $\eta= (\ln \ln n)^2/\ln n$. Then, by the previous statement, with probability $1-\eta=1-o(1)$ we have
$$\frac{X_1}{N_R}\geq \frac{X_1}{2k} \geq 1-O(\ln\ln n /\ln n)\frac{1}{\eta}=1-O(1/\ln\ln n), \hbox{ i.e.}$$ 
$$\Pr(|\frac{X_1}{N_R}-1|\geq \epsilon)= O(\frac{\ln\ln n}{\epsilon \ln n}),$$ 
hence for all  $ \eps$ we have $\lim_{n \rightarrow \infty} Pr[ |\frac{X_1}{N_R}  -l_1|\geq \eps]=0 $, therefore $plim_{n \rightarrow \infty} \frac{X_1}{N_R} =1$.
\end{proof}

\begin{lemma}\label{xi}
In the Reservoir $R$:
\begin{itemize}
\item
for $i\geq 2$, the expected number of nodes of degree $i$ satisfies 
 $\E (X_i)=x_i \leq \frac{(\alpha\sqrt{c})^i } {i!.i-1} . \sqrt{c.n}.$
 \item
 The limit in probability of the ratio $X_i/N_R$ exists and is equal to $\ell_i=0$: that is, $plim_{n \rightarrow \infty} {X_i}/{N_R} =0$.
 \end{itemize}
\end{lemma}

\begin{proof}
For the first statement, consider that a vertex that has degree $j$ in $G$ has degree $i$ in $R$ if and only if the Reservoir picks exactly $i$ of its $j$ edges; since there are $\lfloor c. n/j^2\rfloor$ vertices of degree $j$ in $G$, we can write:
$$x_i= \sum_{j= i}^{\sqrt{c.n}} ~\lfloor \frac{c.n}{j^2}\rfloor .{j \choose i}. (\frac{k}{m})^i. (1-\frac{k}{m})^{j-i} $$
Since $(1-\frac{k}{m})^{j-i}\leq 1$ and ${j\choose i}\leq j^i/i!$, 
$$x_i \leq  c.n.(\frac{k}{m})^i .\frac{1}{i!} \sum_{j= i}^{\sqrt{c.n}} ~\frac{1}{j^2}.{j^i}.$$
We use 
$\sum_{j= i}^{\sqrt{c.n}} {j^{i-2}}\leq \int_0^{\sqrt{cn}}t^{i-2}dt=\frac{1}{i-1}\sqrt{cn}^{i-1}$. 
$$x_i \leq  c.n.(\frac{\alpha}{\sqrt{n}})^i .   \frac{1}{i!(i-1)}\sqrt{cn}^{i-1} 
=    \frac{ ({\alpha}\sqrt{c})^i}{i!(i-1)}\sqrt{cn},$$
hence the first statement.
The second statement follows from $X_i\leq N_R-X_1$ and the fact that $plim_{n\rightarrow\infty}X_1/N_R=1$. 
\end{proof}

\begin{lemma}\label{maxvertexdegreeR}
Let $\eta>0$. Let $i^*$ denote the maximum vertex degree in the Reservoir $R$. Then:
\begin{itemize}
\item
With probability at least $1-\eta$, we have $i^*\leq \log (n/\eta) / \log\log (n/\eta)$.
\item
$\E(i^*)=O(\log n / \log\log n)$.
\end{itemize}
\end{lemma}

\begin{proof}
Consider the first statement. For any $i$ we have $\Pr(i^*<i)=1-\Pr(\exists v \text{ of degree }\geq i)$. Let $i\geq 2\alpha\sqrt{c}$. With the union bound, we compute:
$$\Pr(\exists v \text{ of degree }\geq i)\leq \sum_{i\leq j'\leq\sqrt{cn}} \lfloor \frac{cn}{j^2}\rfloor \sum_{i\leq j\leq j'} {j'\choose j} (\frac{k}{m})^j(1-\frac{k}{m})^{j'-j}.$$
A short calculation ensues: ${j'\choose j}\leq j'^j/j!$, and $1-k/m\leq 1$. Exchanging the order of summation, the right hand side is bounded by
$$\sum_{j=i}^{\sqrt{cn}}  \frac{cn(k/m)^j}{j!}  \sum_{j'=j}^{\sqrt{cn}} j'^{j-2}
\leq  
\sum_{j=i}^{\sqrt{cn}} \frac{cn(k/m)^j}{j!}  \frac{1}{j-1} \sqrt{cn}^{j-1}.$$
Substituting $k/m=\alpha\sqrt{c}$, the right-hand-side is bounded by
$$\sqrt{cn}\sum_{j= i}^{\sqrt{cn}}\frac{1}{j-1}\frac{(\alpha\sqrt{c})^j}{j!}.$$
Since $j\geq i\geq 2\alpha\sqrt{c}$, we have $\alpha\sqrt{c}/j\leq 1/2$, and so we can use bound the sum by the first term times $(1+1/2+1/4+\cdots ) \leq 2$:
$$\sqrt{cn}\sum_{j= i}^{\sqrt{cn}}\frac{1}{j-1}\frac{(\alpha\sqrt{c})^j}{j!}\leq 2\sqrt{cn}\frac{1}{i-1}\frac{(\alpha\sqrt{c})^i}{i}.$$
Fix $\eta>0$. The right-hand side is less than $\eta$ for $i=\frac{\lg(n/\eta)}{\lg\lg (n/\eta)}$.

Consider the second statement. We write 
$$\E(i^*)=\sum_{i\geq 1}\Pr(i^*\geq i)\leq i_0+\sum_{i\geq i_0}\Pr(\exists v \text{ of degree }\geq i).$$
For any $i_0\geq 2\alpha\sqrt{c}$, again we can write
$$\E(i^*)\leq i_0+  \sum_{i\geq i_0}  \sqrt{cn}\sum_{j= i}^{\sqrt{cn}}\frac{1}{j-1}\frac{(\alpha\sqrt{c})^j}{j!}  
\leq i_0+ 2 \sum_{i\geq i_0}  \sqrt{cn}\frac{1}{i-1}\frac{(\alpha\sqrt{c})^i}{i!} $$
$$\leq i_0+ 4\sqrt{cn}\frac{1}{i_0-1}\frac{(\alpha\sqrt{c})^{i_0}}{i_0!} 
$$
Minimizing the right-hand-side over $i_0\geq 2\alpha\sqrt{c}$ gives $\E(i^*)=O(\lg n/\lg\lg n)$.
\end{proof}

\begin{lemma}\label{uniformconvergence} (Uniform convergence)
$$plim_{n\rightarrow\infty} \sum_{i\geq 2}\frac{ i(i-2)X_i}{N_R}=0.$$
\end{lemma}

\begin{proof}
We must prove that for all $\epsilon >0$, we have $ \Pr(\sum_{i\geq 2} i(i-2){X_i}/{N_R}>\epsilon)=o(1)$.

Let $\eta>0$. 
We first use $N_R\geq X_1$ to infer $i(i-2)X_i/N_R\leq (i(i-2)X_i/X_1$. 
\begin{itemize}
\item
By Lemma~\ref{x1a}, with probability at least $1-\eta$ we have $X_1\geq 2k (1-(1/\eta)\ln\ln n/\ln n)$. 
\item
By Lemma~\ref{maxvertexdegreeR}, with probability at least $1-\eta$, the maximum degree $i^*$ satisfies $i^*=O(\log n/\log\log n) $. 
\item
Let $2\leq i\leq i^*$. By Markov's inequality, with probability at least $1-\eta$ we have:\\
 $\sum_{i\geq 2} i(i-2)X_i\leq \sum_{i\geq 2}  i(i-2)x_i/\eta$. Using  Lemma~\ref{xi}, this implies:\\
  $\sum_{i\geq 2} i(i-2)X_i\leq \sum_{i\geq 2}  \frac{(\alpha\sqrt{c})^i } {(i-1)!} . \sqrt{c.n}\frac{1}{\eta}
\leq \alpha\sqrt{c}e^{\alpha\sqrt{c}}\sqrt{cn}/\eta$.
\end{itemize}
Combining, with probability $1-O(\eta)$, all the above statements hold, implying:
$$\sum_{i\geq 2} i(i-2)\frac{X_i}{N_R}\leq   \alpha\sqrt{c}e^{\alpha\sqrt{c}}\sqrt{cn} \frac{1}{\eta} \frac{1}{2k(1-(1/\eta)\ln\ln n/\ln n)} $$
Since $\alpha, c=\Theta(1)$ and $k=\theta(\sqrt{n}\log n)$, the above equation means that 
with probability at least $1-O(\eta)$:
$$\sum_{i\geq 2}i(i-2)\frac{X_i}{N_R}\leq \frac{O(1)}{\eta\log n-\log\log n}
$$
%
Let $\eta=(1/\epsilon+\log\log n)/\log n=o(1)$. Then 
$$\Pr(\sum_{i\geq 2} i(i-2)X_i/N_R \leq O(\epsilon))=1-O(\eta)=1-o(1)$$
\end{proof}

\begin{lemma}\label{averagedegreeR}
In the Reservoir $R$, with probability at least $1-o(1)$, the average vertex degree is at most 2. Here we take expectations over the Configuration.last model, and then compute the average degree over all vertices in the Reservoir.
\end{lemma}

\begin{proof}
The average vertex degree in $R$ equals $2k/N_R$. The probability that it exceeds 2 equals the probability that $N_R<k$, which is at most the probability that $X_1<k$, which by Lemma~\ref{x1a} is $O(\log\log n/\log n)=o(1)$. 
\end{proof}

We now consider the last condition for a well-behaved degree sequence: the limit $L({\mathcal D_R})=plim_{n\rightarrow\infty} \sum_i  i(i-2) X_i/N_R$ exists, is equal to $-1$, and  the convergence is uniform:  $\forall\epsilon\exists i_0 \exists N \forall n>N$ ~~$
|\sum_{i\leq i_0}  i(i-2) X_i/N_R-L({\mathcal D_R})|<\epsilon$. We will take $i_0=2$.

\begin{lemma}\label{uniformlimit} (Uniform limit)
The limit $L({\mathcal D_R})=plim_{n\rightarrow\infty} \sum_i  i(i-2) X_i/N_R$ exists and is equal to $-1$. Moreover, the convergence is uniform: 
Let $\eta>0$. Then, for every $n>N$, with probability at least $1-\eta$ we have:
$$ ~~~|\sum_{i=1}^{i=2}  i(i-2)\frac{X_i}{N_R} - L({\mathcal D_R})|<\frac{O(\ln\ln n)}{\eta \log n}=o(1) $$
$$\text{ and }~
 |\sum_{i\geq 3}  i(i-2)\frac{X_i}{N_R} |<\frac{O(1)}{\eta\log n -\log\log n}=o(1)$$
\end{lemma}

\begin{proof}
This follows from Lemmas~\ref{x1a} and ~\ref{uniformconvergence}.
\end{proof}


\subsubsection{High probability analysis for sequences of degree distributions}

We now consider a sequence of degrees indexed by both $i$ and $n$, where for each $n$: we pick a random degree sequence in $R$ obtained by the second step of the Configuration.last process. We use the following construction to couple the degree sequences as $n$ varies.

Given $ n$ and a degree sequence in the Configuration.last process, let us write $i^*$ for the maximum degree and $ (X_i^{(n)})$    for the degree sequence   $ (X_1^{(n)},X_2^{(n)},... X_{i^*}^{(n)})$.  We define a function $d$ whose value is small with high probability and whose is value is $1$ when two of the Molloy-Reed conditions  are not satisfied: the maximum degree is greater then $N_R^{1/9}$ and the average degree $2k/N_R $ is greater than $2$.

$$d((X_i^{(n)}))= \left\{ \begin{array}{ll} 
1 & \hbox{if } i^*>N_R^{1/9}\\
  & \hbox{or } 2k/N_R > 2\\
\max(|\frac{X_1^{(n)}}{N_R}-1|,|\sum\limits_{i\geq 2} i(i-2)\frac{X_i^{(n)}}{N_R}|) & \hbox{otherwise.}
\end{array}\right. $$
In lemma \ref{uniformlimit}, we separated the two expressions 
$|\frac{X_1^{(n)}}{N_R}-1|$ and $|\sum\limits_{i\geq 2} i(i-2)\frac{X_i^{(n)}}{N_R}|)$ whose limits are small.
For each $n$, we sort the degree sequences in increasing order according to $d$, producing an order on those degree sequences, $(X_i^{(n,1)}), (X_i^{(n,2)}), \ldots$. 

To each number $0<t<1$, we associate the unique sequence $(X_i^{(n,{r})})$ such that $\sum_{j<r}\Pr((X_i^{(n,j)}))<t<\sum_{j\leq r}\Pr((X_i^{(n,j)}))$. Note that $r$ is a function of $t$ and $n$. This provides the desired coupling:  for each $t$, as $n$ spans the natural integers we obtain a collection of degree sequences, one for each $n$. 
We wish to apply the Molloy-Reed Theorem (Theorem~ \ref{mr:bound}) to the resulting sequence of degree sequences $(X_i^{(n,r)})$ and conclude.\\

{\bf Proof of theorem \ref{t1}}.
Consider a sequence indexed by both $i$ and $n$, generated as above. We check that this asymptotic degree sequence is almost surely well-behaved and the other three assumptions of the Molloy-Reed Theorem hold.

First, it is well-behaved: it is feasible by construction. The smoothness condition, i.e. the existence of 
$plim_{n \rightarrow \infty} {X_i^{(n,r)}}/{N_R}$
 is proved  in the 
lemma \ref{x1a} for $i=1$ and in the lemma \ref{xi} for $i>1$. The  uniform convergence of
$plim_{n\rightarrow\infty} \sum_{i\geq 2}\frac{ i(i-2)X_i^{(n,r)}}{N_R}$ is proved in the lemma
 \ref{uniformconvergence}.
 The existence and uniform convergence of $\sum_{i\geq 1}  i(i-2)\frac{X_i^{(n,r)}}{N_R}$
 is proved in the lemma
  \ref{uniformlimit}.
   The condition on the maximum degree, less than $N_R^{1/9}$, is proved in the lemma\ref{maxvertexdegreeR}  and the condition on the average degree in the lemma    \ref{averagedegreeR}.
  The coeffcient $Q$ is $-1$ and we can therefore apply the Molloy-Reed theorem \ref{mr:bound}.
 The conclusion, i.e. the bound on the size of the largest connected component $C$ holds almost surely. There exists a constant $B$ such that:
 $Pr_t[ |C| \leq B. N_R \leq B. k^{1/4}]=1$.
Hence:
$$  Pr_{{\rm Configuration.last}}[ |C| \leq B.k^{1/4}]=
Pr_{\mu \cdot \Omega}[ |C| \leq k^{1/4}]\rightarrow_{n\rightarrow\infty}1~~$$
The algorithm \Aone tests  if  $|C| \leq k^{1/4}\leq   n^{1/8}\log^2 n   $, hence it is correct with h.p.$\Box$

 \section{Space lower bounds}
We reduce a hard problem for the 1-way communication complexity to the existence of an $(1,\delta)$-large very dense subgraph and assume the reader familiar with the concepts of communication complexity \cite{KN97}. The multiparty disjointness problem in the 1-way communication model is defined as follows.
 There are $q$ players and for each $j=1,...q$, player $j$ has an $n$-bit vector $x_j=x_{j,1}\ldots x_{j,n}$. In the restricted problem:
 \begin{itemize}
 \item
  either all vectors $x_j$ are pairwise distinct (i.e. there is no $j,j',i$ such that $x_{j,i}=x_{j',i}=1$), 
  \item
  or there exists a unique $i^*$ such that $\bigwedge_j  x_{j,i^*}=1$, and all vectors $x_j$ are otherwise pairwise distinct (i.e. there is no $j,j',i \neq i ^*$ such that $x_{j,i}=x_{j',i}=1$).
  \end{itemize}
In the 1-way communication model, information is only sent from a player $j$ to a player $j'$ such that $j'>j$, and the last player, player $q$, must decide whether there is an $i$ such that $\bigwedge_j  x_{j,i}=1$. 

\begin{lemma}\label{lemma:cha}\cite{Cha03}
The restricted multiparty disjointness problem with 1-way communication and $q$ players requires communication complexity 
$\Omega(n/q)$.
\end{lemma}
This lower bound is used to obtain other lower bounds for a range of problem.
An $\Omega(n)$ lower bound is presented in \cite{Bah12} for the $\alpha$-approximation of the maximum density ratio $\rho^*$, i.e. to find a subgraph with a vertex $S$ such that $\rho(S) \geq \rho^*/\alpha$.

\begin{lemma}\label{lw}\cite{Bah12}
An $\sqrt{q/2}$-approximation streaming algorithm for the maximum density ratio requires 
 space $\Omega(n/q)$.
\end{lemma}

We keep the same  reduction presented in \cite{Bah12}: it  reduces the  restricted multiparty disjointness problem with 1-way communication and $q$ players to the $\sqrt{q/2}$-approximation  of the maximum density ratio  $\rho^*$.   Consider an instance of the restricted  $q$-party disjointness problem with 1-way communication. Player $j$ holds boolean variables $x_{j,1},x_{j,2},\cdots ,x_{j,n}$, for each $j=1,2,\ldots ,q$. We construct a graph $G$ such that:\\

$\bullet$ $G$ is the union of $n$ disjoint graphs $G_1,...G_n$, each over $q$ vertices. For $i=1,2,\ldots ,n$, the nodes of $G_i$ are denoted $u_{1,i},u_{2,i},...u_{q,i}$. 

$\bullet$ If $x_{j,i}=1$, we add the $q-1$ edges from the node $u_{j,i}$ to all the other nodes  $u_{j',i}$ of $G_i$, for $j'\neq j$.  \\

For a Yes instance of the multiparty disjointness problem, i.e. $\bigwedge_j  x_{j,i^*}=1$, the graph $G_{i^*}$ is a clique  of size $q$ and the maximum density ratio is therefore $\rho^*=(q-1)/2$. For a No instance, $G$ is a forest where each tree is of depth $1$ and the maximum density ratio is $\rho^*=(q-1)/q=1-1/q.$ 

Let an input stream with the edges of player 1, then of player 2, ... then of player $q$ coming last. A streaming algorithm for $\rho^*$ with approximation less than $\sqrt{q/2}$ and space $o(n/q)$ could decide between a Yes and a No instance of the $q$-multiparty disjointness problem.
If we assume that $S$ must be large, let $q=\delta \sqrt{n}$ and we obtain:

\begin{corollary}\label{lowerbound}
A $n^{1/4}$-approximation streaming algorithm for  the maximum density ratio for  $S$ of size at least $2\sqrt{n}$ requires space 
$\Omega(\sqrt{n})$.
\end{corollary}

If we consider the very dense criterium, i.e. $\gamma=1$, we obtain our lower bound:
\begin{corollary}\label{corollary:lowerbound}
The detection of a $(\gamma,\delta)$ large very dense subgraph requires $\Omega(\sqrt{n})$ space.
\end{corollary}
\begin{proof}
Let $q=\delta \sqrt{n}$ in the  reduction of  Lemma \ref{lw} in \cite{Bah12}.   A  Yes (resp. No) instance of
 the multiparty disjointness problem is reduced to a Yes (resp. No) instance of the $(1, \delta)$-Large very dense subgraph problem. It implies the space lower bound of \ref{lemma:cha}.
\end{proof}

\section{Concentrated case: reconstructing a complete subgraph}\label{Sestimate}
In this section we propose an algorithm (Algorithm 2) to approximately reconstruct a very dense subgraph. We show that in the Concentrated model with a clique $S$ (a   \gdense with $\gamma=1$), the output approximates $S$ with high probability.

To reconstruct a clique, based on Algorithm \Aone, it would be tempting to output the largest connected component in the Reservoir. However, in the Concentrated model with a clique $S$, such an algorithm would overestimate the size of $S$ and output many vertices besides the vertices of $S$.
Instead, we observe that many of the extraneous vertices appear to be leaves of the connected component of the Reservoir, and can thus be eliminated by outputing the $2$-core instead of the full connected component. Formalizing this intuition leads to the following algorithm and result.

{\bf Static very dense subgraph estimation  Algorithm 2}.
{\em Let $C$ be the  largest connected component of the reservoir $R$. If $|C| < n^{1/8}\ln^2 n$ then Reject, else output $\twocoreC$}.

We now assume that the parameter $\alpha$ depends on $n$ and slowly grows to infinity. For example $\alpha=O(\ln \ln \ln n)$ so that the analysis of the previous sections still hold. If $\alpha$ is constant, we  just approximate $S$ within a constant factor.

\begin{theorem}\label{thm:estimation}
Assume the concentrated model (Section~\ref{subsection:uniform}) with a clique $S$, and let $\widehat{S}=\twocoreC$ denote the output of Algorithm 2.  If $
1/\delta=o(\alpha )$ and $\alpha=O(\ln \ln \ln n/\sqrt{c})$, then  $|S\setminus \widehat{S}|$ and $|\widehat{S}\setminus S|$ almost surely are both $o( |S|)$.
\end{theorem}

Since $S$ is a clique in this section, the edges of the reservoir $R$ that are internal to $S$ are exactly distributed according to $G (n,p)$. Notice that:

\begin{itemize}
\item If $v \in S$, then the degree of $v$ is $O(\sqrt{n})$

\item The degree of $v$ in the Reservoir is $\log n$
\end{itemize}

Recall Theorem \ref{er:bound} which gives us an estimate  on the size of $ |C|$ ($ |C|\simeq b.n$) and of its
$\twocoreC$ ($  |\twocoreC|\simeq b.(1-t).n$).
If $S$ is a clique, we can use this result and conclude that the $  \twocoreC$ is of size $O(n)$ and in $S$ with high probability.  The proof follows the successive steps, where results are taken with high probabilities:

\begin{itemize}
\item The Reservoir has no cycle disjoint of $S$, Lemma \ref{d1} ,

\item The $\twocoreC$ consists of elements of $S$ and possibly elements of $V-S$ which belong to cycles that go through $S$,

\item Lemmas \ref{d2} and \ref{d3} show that very few nodes of $V-S$ belong to such cycles. 
\end{itemize}

 
\begin{lemma}\label{d1}
With high probability, in the restriction of the Reservoir to $V-S$, the total size of connected components that have a cycle is at most $O(n^{1/4})$. \end{lemma}
\begin{proof}
Use Molloy-Reed, Theorem~\ref{mr:bound}.
\end{proof}

Nodes of $S$ have a high degree. Some nodes in $V-S$ have also a high degree. Let us study false positives for $S$, that is, vertices $u$ that belong to  $\twocoreC$ but not to $S$. $$ \mu=\Pr[(u\in V-S \wedge u \in   \twocoreC) ]   $$
One possibility is that $u$ is connected to $\twocoreC\cap S$ by two disjoint paths. We first consider the case where those paths have length 1.

\begin{lemma}\label{d2}
Let $Z$ denote the number of vertices $u$ in $V-S$ incident to two edges of $R$ into $S$,  $(v_1,u)\in R$ and $(v_2,u)\in R$ with $v_1,v_2,\in S$. Then:
$$\E[ Z  ]  \leq \sqrt{n}\frac{O(\alpha^2)}{ (\log n) ^2}$$
\end{lemma}
\begin{proof}
By linearity of expectation, the desired quantity is equal to: 
$$A=\sum_{u\in V-S} \Pr[ \hbox{there exist } v_1,v_2\in S \hbox{ such that } (v_1,u),(u,v_2) \in R ]   .$$
Edges $(v_1,u)$ and $(v_2,u)$ are border edges of $S$. By definition of $S$-concentrated dynamics and of the size of $R$, $S$ has $O(\alpha \sqrt{n})$  border edges in $R$. By construction of the power law graph, each of those edges is attached to a random stub of $V-S$. Since there are $\sim cn\log n/2$ stubs in $V-S$, the probability that such an edge is attached to a stub of $u$ is $\Theta(d_u/(n\log n))$. We will also use
$$\sum_{u\in V-S}  d_u^2 \sim \sum_{i=1}  ^{i=\sqrt{cn}}c.n .i^2/i^2 =\sum_{i=1}  ^{i=\sqrt{cn}}c.n =c\sqrt{c}.n.\sqrt{n}=\Theta(n^{3/2}).$$

$$A =\sum_{u\in V-S} { O(\alpha\sqrt{n}) \choose 2} (\Theta(d_u/n.\log n))^2 =\frac{O(\alpha^2)}{n. (\log n )^2}\sum_{u\in V-S}  d_u^2 $$
$$=  \frac{O(\alpha^2)\sqrt{n}}{( \log n )^2}$$
\end{proof}

\begin{lemma}\label{d3}
Given $i\geq 2$, let $Y$ denote the number of chains $(v,u_1,u_2,\ldots, u_i, v')\in S\times (V\setminus S)\times \cdots \times (V\setminus S)\times  S$ of $R$. 
Then:
$$\E[ Y]  \leq \sqrt{n} \frac{O(\alpha^{i+1})}{ (\log n) ^{i+1}}.$$

\end{lemma}
\begin{proof}
We will give the full proof for $i=2$, and the generalization to $i\geq 2$ is immediate. Fix two nodes $u_1,u_2 \in V-S$.
$$ \Pr[ (u_1,u_2)\in R] = \Pr[ (u_1,u_2)\in G] \Pr[(u_1,u_2)\in R|(u_1,u_2)\in G  ]$$
$$\sim \frac{d_{u_1}.d_{u_2}}{m }.\frac{O(\alpha)}{\sqrt{n} }$$
For $u_1$, the probability to be connected to $S$ is:
$$ \Pr[ \exists v\in S: (v,u_1)\in R] =  \frac{ d_{u_1}}{m }O(\alpha)\sqrt{n}.$$
Similarly for $u_i$,
$$ \Pr[ \exists v'\in S: (u_i,v')\in R] =  \frac{ d_{u_i}}{m }O(\alpha)\sqrt{n}.$$

Combining, the probability that there exists $v,v'\in $ with a path $(v,u_1,u_2,v')$ in $R$ is bounded by 
$$\Pr(\exists v,v'\in S: (v,u_1,u_2,v')\hbox{ path in }R)= \frac{d_{u_1}^2d_{u_2}^2}{m^3}O(\alpha^3){\sqrt{n}}.$$
Summing over $u_1,u_2\in V-S$ and recalling that $\sum_u d_u^2=O(n^{3/2})$ and that $m=\Theta(n\log n)$, we obtain that the expected number of chains, for $i=2$, is
$$\sqrt{n} \frac{O(\alpha^3)}{(\log n)^3}.$$

The generalization to $i>2$ is straightforward.

\end{proof}

 \begin{proof}(Proof of Theorem~\ref{thm:estimation})
 We use Theorem~\ref{er:bound} to prove that $\twocoreC$ contains a large fraction of the elements of $S$.
 Indeed, the theorem means that $|\twocoreC\cap S|\geq |S| b(1-t)$ 
  with $b=1-t/c_1$ and $te^{-t}=c_1e^{-c_1}$, where the probability that an edge of $S$ is put in the Reservoir is $c_1/|S|$. We have $|S|=\delta \sqrt{n}$ and $\alpha/\sqrt{n}=c_1/|S|$. Hence $c_1=\alpha.\delta$. 
 
 Assume $\alpha >> 1/\delta$.  Then $c_1\rightarrow\infty$, so $t\rightarrow 0$ and $b\rightarrow 1$, and then $ |S| b(1-t)\sim |S|$, so $|S\setminus \widehat{S}|=o(|S|)$. 
 
 To analyze $|\widehat{S}\setminus S|$, observe that the nodes of $V-S$ that belong to $\twocoreC$ either  have two disjoint paths leading to $S$, or belong to a cycle of $V-S$.
 
 The nodes that belong to a cycle of $V-S$ are few in number, $O(n^{1/4})$ by Lemma~\ref{d1}.
 
 Concerning the nodes that have two disjoint paths leading to $S$,  Lemmas~\ref{d2} and~\ref{d3} give: 
 The number of nodes $u$ in $ V-S$  that have two disjoint paths leading to $S$ is at most: 
$$\sum_{i\geq 1} i \sqrt{n}\frac{O(\alpha^{i+1})}{(\log n)^{i+1}}\simeq \sqrt{n}\frac{O(\alpha^2)}{(\log n)^2}.$$

Thus the expected total number of nodes of $V-S$ that belong to $\twocoreC$ is  $o(|S|)$.  
 \end{proof}
 
\section{Dynamic  graphs}

Consider the sequence of graphs $G_i$ defined by the edges in each window $w_i$.  We keep a Reservoir $R_i$ but only store the $2$-core of the large connected components. We extend our model of random graphs with or without a very dense subgraph to the dynamic case.

Let $P$ be the graph property:  there is a   \gdense of size greater than $\delta.\sqrt{n}$. How do we decide $\Diamond ~P(t)$, i.e. there is a window $w_i$ at some time $t_i \leq t$ such that  $G_i$ has a   \gdense of size greater than $\delta.\sqrt{n}$? Recall that $\alpha=\Theta(1/\gamma.\delta)$ as in the static case.

{\bf Dynamic very dense subgraph detection  Algorithm 3 $(t)$}: {\em let $C_i$ be the  largest connected component of the reservoir $R_i$ of size  $k=\Theta(\alpha\sqrt{n}\log n)$ at time $t_i \leq t$. If  there is an $i$ such that $|C_i| \geq  n_i^{1/8}\ln^2 n_i$ Accept, else Reject.}

Consider the following {\bf Dynamics} applied to a given graph $G$: remove $q \geq 2$ random edges, uniformly on the set of edges of $G$, freeing $2.q$  stubs.  In the case of the {\bf uniform Dynamics}, we generate a new uniform matching on these free stubs to obtain $G'$. 

In the case of the {\bf $S$-concentrated Dynamics}, we have fixed some  subset $S$  of size $\delta.\sqrt{n}$ among the nodes of high degre and some $\gamma$. 
Consider the $q$ edges that we remove. Partition then into internal edges 
$E(S)$,  external edges in $E(\bar{S})$ and  cut edges of $ E(S,\bar{S}). $ 
 There are three cases and
the analysis generalizes the static case.


\subsection{Dynamic models}\label{section:dynamic}

For  the {\bf $S$-concentrated Dynamics}, we have fixed some  subset $S$  of size $\delta.\sqrt{n}$ among the nodes of high degre and some $\gamma$. Partition the  $q$ removed edges into internal edges 
$E(S)$,  external edges in $E(\bar{S})$ and  cut edges of $ E(S,\bar{S}). $ 
 There are three cases:
\begin{enumerate}
\item The graph $G$ had a   \gdense subgraph and we will maintain such a very dense subgraph.
We rematch the corresponding stubs of each class with a uniform matching. We conserve the same number of  edges in  $E(S)$,   $E(\bar{S})$ and $ E(S,\bar{S})$ and  maintain a   \gdense subgraph. 
\item  The graph had no   \gdense subgraph.  Let  $q'$ be the number of edges in $ E(S,\bar{S})$. With $\lfloor q'/2 \rfloor$ new edges in $S$, we don't have a   \gdense subgraph.
If $q'$ is even, we match uniformly all the stubs in $S$. If $q'$ is odd, we only take
$q'-1$ stubs, match  uniformly all the stubs in $S$ and leave one edge in $ E(S,\bar{S}) $.
\item  The graph had no   \gdense subgraph but with $q'' < \lfloor q'/2 \rfloor$ edges, we reach a 
  \gdense subgraph. 
 We take the $q''$ stubs in $ S$, match uniformly all the
 stubs in $S$, and match the other edges in $ E(S,\bar{S}) $ uniformly.

\end{enumerate}
How does the distribution of random graphs evolve in time? Consider the Markov chain $M$ where nodes are the possible graphs and transition probabilities $M(i,j)$ are the probabilities to obtain $G_j$ from $G_i$ with the process of removing random $q$ edges from $G_i$ and recombining the $2q$ stubs with the uniform or the  $S$-concentrated Dynamics.
\begin{lemma}\label{uniformlemma}
For the uniform Dynamics  the stationary distribution of the Markov chain $M$ is uniform. For the $S$-concentrated $S$-Dynamics,  the stationary distribution of the Markov chain $M$ is uniform among all the graphs with a   \gdense $S$. 
\end{lemma}
\begin{proof}
By definition of the Markov chain, for all $i$, $\sum_j M(i,j)=1$. The transition from $G_i$ to $G_j$ can also occur backwards from $G_j$ to $G_i$, and similarly  for many  $G_k$ which lead to  $G_i$. Therefore for all $j$, $\sum_i M(i,j)=1$ as it sums all possible transitions starting from $G_j$. The matrix $M$ is then doubly stochastic and ergodic. Hence the stationary distribution is uniform. Fo the $S$-concentrated dynamics, we reach a   \gdense subgraph with the right number of edges. Each transformation of a   \gdense subgraph into another one is reversible. Hence the matrix $M$ is also doubly stochastic and the stationary distribution is uniform.
\end{proof}

%
%

\subsection{Deciding properties of dynamic random graphs}\label{a:dynamic}

A {\bf general Dynamics} is a function which chooses at any given time, one of the two strategies: either a uniform Dynamics or an $S$-concentrated Dynamics for a fixed $S$.
An example is the  {\bf Step Dynamics}: apply the uniform Dynamics first, then switch to the $S$-dynamics for a time period $\Delta$, and switch back to the uniform Dynamics. Notice that during the uniform Dynamics phase, there are no large components in the Reservoir. For the step phase, we store some components which will approximate $S$.

A stream $G(t)$ has a large 
   \gdense subgraph  if there is $G_i$ and  $S_i$ such that  $|S_i| > \delta.\sqrt{n}$ and
$\alpha > \frac{ (1+\epsilon)}{\gamma .\delta}$ for a $t_i \leq t$.
Assume $Prob[ {\rm ~~Algorithm~ ~1 ~~Accepts~~} ] \geq 1-\eta$ in Theorem \ref{tConcentratedSize}. We now show that the detection error  of Algorithm 3  decreases.

\begin{corollary}\label{t0d}
  For a stream  which  has a large   \gdense (which satisfies the conditions of Theorem 
\ref{tConcentratedSize})  during a time  interval  $\Delta \geq \tau$, Algorithm 3  is such that
$Prob[ {\rm ~~Algorithm~ ~3 ~~~Accepts~~} ] \geq 1-\eta^{ \Delta/\tau }$.

\end{corollary}
\begin{proof}
There are $\lfloor (\Delta-\tau)/\lambda \rfloor $ different windows but $\lfloor  \Delta/\tau \rfloor $ independent windows, i.e. windows with no overlap. The samples are then independent and we can then amplify the success probability. 
\end{proof}


For the random graphs  generated  by the uniform Dynamics, the situation is different and the error will increase.  If we assume that:
$Prob[ {\rm ~~Algorithm~ ~1 ~~Rejects~~} ] \geq 1-\eta '$, the detection error
of Algorithm 3 is given by:

\begin{corollary}\label{t0duni}
  For a stream  of random graphs which follow the uniform Dynamics, Algorithm 3  is such that:
$$Prob[ {\rm ~~Algorithm~ ~3 ~~~Rejects~~} ] \geq 1- (\frac{t-\tau}{\lambda}) .\eta '$$

\end{corollary}
\begin{proof}
There are $\lfloor \frac{t-\tau}{\lambda}\rfloor $  windows and the error probability is less than the sum of the errors for each window.
\end{proof}

\subsection{Dynamic estimation of $S$}

Consider random graphs generated by a step Dynamics, i.e. a strategy which maintains an
$S$ concentration during   a time  interval  $\Delta \geq \tau$. In this case, we can improve the quality of the approximation of Theorem \ref{thm:estimation}. Assume  $|S\setminus \widehat{S}|\leq \rho.|S|$ for Algorithm 2, where $\rho \leq 1$.  Let $\eps$ be an arbitrary  tolerated  error, and say that $ \widehat{S}$  $\eps$-approximates $S$  if  $|S\setminus \widehat{S}|\leq \eps.|S|$  \\

{\bf Dynamic  very dense subgraph estimation  Algorithm 4 (t, $\eps$)}.
{\em Let $\widehat{S}=\emptyset$. Consider the first independent $I= \log(\eps) /\log(\rho)$  windows $w_i$ when $C_i$,  the  largest connected component of the reservoir $R_i$  at time $t_i <t$ is such that
 $|C_i| > n_i^{1/8}\ln^2 n_i$. Then  $\widehat{S}=\widehat{S}\cup 2\hbox{-core}(C_i)$}.

\begin{corollary}\label{t1d}
  For a stream  generated by the step Dynamics  during a time  interval  $\Delta \geq I.\tau$, 
Algorithm 4  will $\eps$-approximate  $S$ almost surely.

\end{corollary}
\begin{proof}
There are at least $I$ independent windows where we have a large connected component. 
For each element in $S$, there is a probability $\rho$ not to be selected in $\widehat{S}$ for each window.
For $I$ independent windows, the probability not to be selected is $\rho^I= \eps$, hence a point
is selected with probability $1-\eps$. By Theorem  \ref{thm:estimation} very few nodes of $V-S$ are in $\widehat{S}$.
\end{proof}

\subsection{Implementation}

An implementation of the method,  in \cite{RV18}, considers streams of edges  defined  by Twitter graphs\footnote{The nodes of a Twitter graph are the tags, either @x or \#t. A tweet sent by @x which contains the tags @y and \#t generates the edges $(@x,@y)$ and  $(@x,\#t)$. Given a tag or a set of tags, Twitter sends  in a stream all the tweets which contain one of the selected tags. The stream of tweets is transformed into a  stream of edges.} associated with tags (for example \#bitcoin or \#cnn).  With windows of length $\tau=1$ hour and  step $\lambda=30$ mins, there were  $m=20.10^3$ edges per window and the size $k$ of the Reservoir was $500$, of the order of $\sqrt{m}$. The \#bitcoin stream has a unique  giant component $C$ in the Reservoir with  approximately $100$ edges, i.e. a  compression factor of approximately $2.10^2$. In practice, a giant component in the Reservoir is the witness of a large $\gamma$-clique,  even though the coefficient $\gamma$ can be small.
 The  analysis of the variations of the sizes of the giant components over time is one of the motivations for the dynamic random graphs introduced in this paper. Classical methods  \cite{Ag10} to detect communities in social graphs use Modules and Spectral techniques, which require to store the entire graph.


\bibliography{x1}



\end{document}